\documentclass[12pt]{article}
\usepackage{amssymb,amsmath,amsthm,amscd,latexsym}
\usepackage{mathrsfs}
\usepackage{mathrsfs}
\usepackage{amsfonts}
\usepackage{amsmath}
\usepackage{amssymb}
\usepackage{amscd}
\usepackage{mathrsfs}
\usepackage{amssymb}
\usepackage{amsmath}
\usepackage{amsthm}
\usepackage{color}
\usepackage{latexsym}
\usepackage{indentfirst}
\usepackage{enumitem}
\usepackage{anysize}
\usepackage{bbm}
\usepackage[normalem]{ulem}
\usepackage{soul}
 \usepackage{cite}
\usepackage{cancel}

\renewcommand{\paragraph}{\roman{paragraph}}

\input cyracc.def

\newcommand{\Z}{\mathbb{Z}}

\newcommand{\G}{\mathbb{G}}
\newcommand{\vv}{\mathbf{v}}
\newcommand{\x}{\mathbf{x}}
\newcommand{\y}{\mathbf{y}}

\theoremstyle{definition}

\newtheorem{thm}{Theorem}

\newtheorem{lem}{Lemma}
\newtheorem{prop}{Proposition}

\newcommand{\C}{\mathcal{C}}

\begin{document}
\title{\bf Rank and pairs of Rank and Dimension of Kernel of $\mathbb{Z}_p\mathbb{Z}_{p^2}$-linear codes
\thanks{This research is supported by the National Natural Science Foundation of China (61672036),
the Excellent Youth Foundation of Natural Science Foundation of Anhui Province (1808085J20), the
Academic Fund for Outstanding Talents in Universities (gxbjZD03).}}

\author{
\small{Xiaoxiao Li$^{1}$, Minjia Shi$^{1}$, Shukai Wang$^{1}$}\\ 
\and \small{${}^1$School of Mathematical Sciences, Anhui University, Hefei, 230601, China}\\
}
\date{}
\maketitle
\begin{abstract}
A code $C$ is called $\Z_p\Z_{p^2}$-linear if it is the Gray image of a $\Z_p\Z_{p^2}$-additive code. For any
prime number $p$ larger than $3$, the bounds of the rank of $\Z_p\Z_{p^2}$-linear codes are given. For each value of the rank and the pairs of rank and the dimension of the kernel of $\Z_p\Z_{p^2}$-linear codes, we give detailed construction of the corresponding codes. Finally, as an example, the rank and the dimension of the kernel of $\Z_5\Z_{25}$-linear codes are studied.
\end{abstract}

{\bf Keywords:} Linear codes, $\Z_p\Z_{p^2}$-linear codes, $\Z_p\Z_{p^2}$-additive codes, kernel and rank
\section{Introduction}
In 1973, Delsarte first defined the additive codes in terms of association schemes \cite{DP}, it is a subgroup of the underlying abelian group.
At first, the research on $\Z_2\Z_4$-additive codes aroused widespread interest. Borges et al. \cite{BFPR10} studied the standard generator matrix and the duality of $\Z_2\Z_4$-additive codes. Dougherty et al. \cite{DLY16} constructed one weight $\Z_2\Z_4$-additive codes and analyzed their parameters. Benbelkacem et al. \cite{BBDF20} studied $\Z_2\Z_4$-additive complementary dual codes and their Gray images. More structure properties of $\Z_2\Z_4$-additive codes can be found in \cite{M92,BBCM,BBDF11,JTCR2}.

Additive codes over different mixed alphabet have also been intensely studied, for example $\Z_2\Z_2[u]$-additive codes \cite{BC2}, $\Z_2\Z_{2^s}$-additive codes \cite{AS13}, $\Z_{p^r}\Z_{p^s}$-additive codes \cite{AS15}, $\Z_2\Z_2[u,v]$-additive codes \cite{SW}, $\Z_p\Z_{p^k}$-additive codes \cite{SWD} and $\Z_p(\Z_p+u\Z_p)$-additive codes \cite{WS}, and so on. It is worth mentioning that $\Z_2\Z_4$-additive cyclic codes form an important family of $\Z_2\Z_4$-additive codes, many optimal binary codes can be obtained from the images of this family of codes. More details of $\Z_2\Z_4$-additive cyclic codes can be found in \cite{BC,JCR,JTCR,JTCR2,YZ20}.

Let $\mathbb{Z}_p$ and $\mathbb{Z}_{p^2}$ be the ring of integers modulo $p$ and $p^2$, respectively, where $p>2$ is prime.
For a positive integer $m$, $\Z_p^m$ and $\Z_{p^2}^m$ are the ring extension.
The linear codes over $\Z_p$ and $\Z_{p^2}$ are subgroups of $\Z_p^m$ and $\Z_{p^2}^m$, respectively.

The usual Gray map from $\Z_4$ to $\Z_2^2$, given in \cite{a15}, has been generalized to a Gray map from $\Z_{2^s}$ to $\Z_{2}^{2^{s-1}}$ in \cite{a7,a19}. Actually, Carlet's Gray map from \cite{a7} is a particular case of the Gray map given in \cite{a19} satisfying $\sum\lambda_i\phi(2^i)=\phi(\Sigma\lambda_i2^i)$ \cite{a11}. Let $Y$ be the matrix of size $(s-1)\times p^{s-1}$ whose columns are all different vectors in $\Z_p^{s-1}$. Carlet's Gray map from $\Z_{2^s}$ to $\Z_{2}^{2^{s-1}}$ \cite{a7} can be generalized to a map from $\Z_{p^s}$ to $\Z_p^{p^{s-1}}$ as follows \cite{a13}:
$$\phi(u)=(u_{s-1},u_{s-1},\ldots,u_{s-1})+(u_{0},u_{1},\ldots,u_{s-2})Y,$$
where $u\in \Z_{p^s}$; $[u_{0},u_{1},\ldots,u_{s-1}]_p$ is the $p$-ary expansion of $u$, that is, $u=\sum_{i=0}^{s-1}u_ip^i$ with $u_i\in \Z_p$. In this paper, we consider the case $s=2$.

Let $\C$ be a $\Z_p\Z_{p^2}$-additive code of length $\alpha+\beta$ defined by a subgroup of $\Z_p^\alpha\times\Z_{p^2}^\beta$. Let $n=\alpha+p\beta$ and
$\Phi:\Z_p^\alpha\times\Z_{p^2}^\beta\rightarrow \Z_p^n$ is an extension of the Gray map, which
$$\Phi(\x,\y)=(\x,\phi(y_1),\ldots,\phi(y_\beta)),$$
for any $\x\in \Z_p^\alpha$, and $\y=(y_1,\ldots,y_\beta)\in \Z_{p^2}^\beta$. The code $C=\Phi(\C)$ is called a $\Z_p\Z_{p^2}$\emph{-linear code}.

There are two parameters of $\Z_p\Z_{p^2}$\emph{-linear code} $C=\Phi(\C)$ which is a nonlinear code: the rank and the dimension of the kernel.
We denote $\langle C\rangle$ the linear span of the codewords of $C$. The dimension of $\langle C\rangle$ is called the \emph{rank} of the code $C$, denoted by $rank(C)$.
The \emph{kernel} of the code $C$, which is denoted as $K(C)$, is defined as:
$$K(C)=\{\x\in\Z_p^{n}|C+\x=C\},$$
where $C+\x$ means $\x$ adds all codewords in $C$.
We will denote the dimension of the kernel of $C$ by \emph{$ker(C)$}. These two parameters are helpful to the classification of $\Z_p\Z_{p^2}$-linear codes.

The rank and dimension of the kernel have been studied for some families of $\Z_2\Z_4$-linear codes \cite{BPR,FPV,PRV}.
In \cite{CFC}, Fern\'{a}ndez et al. studied the rank and kernel of $\Z_2\Z_4$-linear codes.
They also studied the rank and kernel of $\Z_2\Z_4$-additive cyclic codes in \cite{JTCR2}.
In \cite{SWL}, the authors studied the rank of $\Z_3\Z_9$-linear codes and the kernel of $\Z_p\Z_{p^2}$-linear codes.
In this paper, we consider the general case for the rank and kernel of $\Z_p\Z_{p^2}$-linear codes.

The paper is organized as follows.
In Section \ref{sec:2}, we give some properties about $\mathbb{Z}_p\mathbb{Z}_{p^2}$-additive codes and $\mathbb{Z}_p\mathbb{Z}_{p^2}$-linear codes. Section \ref{sec:4} gives an explicit formulation of a special equation which helps us to study the rank of $\mathbb{Z}_p\mathbb{Z}_{p^2}$-linear codes. Then, we find all values of the rank for $\mathbb{Z}_p\mathbb{Z}_{p^2}$-linear codes, and we construct a $\mathbb{Z}_p\mathbb{Z}_{p^2}$-linear code for each value.
In Section \ref{sec:5}, pairs of rank and the dimension of the kernel of $\mathbb{Z}_p\mathbb{Z}_{p^2}$-additive codes are studied.
For each fixed value of the dimension of the kernel, the range of rank is given.
And the construction method of the $\mathbb{Z}_p\mathbb{Z}_{p^2}$-linear codes with pairs of rank and the dimension of the kernel is provided.
As an example, the rank and the dimension of the kernel of $\Z_5\Z_{25}$-linear codes are studied in section \ref{sec:6}. In Section \ref{sec:7}, we conclude the paper.

\section{Background material}\label{sec:2}
Let $\mathcal{C}$ be a $\mathbb{Z}_p\mathbb{Z}_{p^2}$-additive code. Then $\mathcal{C}$ is isomorphic to an abelian structure $\mathbb{Z}_p^\gamma \times \mathbb{Z}_{p^2}^\delta$ since it is a subgroup of $\mathbb{Z}_p^\alpha \times \mathbb{Z}_{p^2}^\beta$.
The \emph{order} of a codeword $\mathbf{c}$ means the minimal positive integer $a$ such that $a\cdot\mathbf{c}=\mathbf{0}$.
Therefore, the size of $\mathcal{C}$ is $p^{\gamma+2\delta}$, and the number of codewords of order $p$ in $\mathcal{C}$ is $p^{\gamma+\delta}$. We denote $\mathcal{C}_X$ is the punctured code of $\mathcal{C}$ by deleting all coordinates over $\mathbb{Z}_p$. Let $\mathcal{C}_p$ be the subcode consisting of all codewords of order $p$ in $\mathcal{C}$. Let $\kappa$ be the dimension of the linear code $(\mathcal{C}_p)_X$ over $\Z_p$.
If $\alpha=0$, then $\kappa=0$.
Considering all these parameters, we will say that $\mathcal{C}$ or $C=\Phi(\mathcal{C})$ is of type $(\alpha, \beta; \gamma, \delta; \kappa)$. The code $\mathcal{C}$ is generated by $\gamma$ codewords $\mathbf{u}_1,\mathbf{u}_2,\ldots,\mathbf{u}_\gamma$ of order $p$ and $\delta$ codewords $\mathbf{v}_1,\mathbf{v}_2,\ldots,\mathbf{v}_\delta$ of order $p^2$, which means every codeword $\mathbf{c}\in \mathcal{C}$ can be uniquely expressed in the form
$$\mathbf{c}=\sum_{i=1}^\gamma \lambda_i\mathbf{u}_i+\sum_{j=1}^\delta \nu_j\mathbf{v}_j,$$
where $\lambda_i \in \mathbb{Z}_p$ for $1\leq i \leq \gamma$ and $\nu_j \in \mathbb{Z}_{p^2}$ for $1 \leq j \leq \delta$, respectively. Then, we get the generator matrix $\mathcal{G}$ for the code $\mathcal{C}$ by words $\mathbf{u}_i$ and $\mathbf{v}_j$.

Recall that in coding theory, two linear codes $C_1$ and $C_2$ of length $n$ are \emph{permutation equivalent}
if there exists a coordinate permutation $\pi$ such that $C_2=\{\pi(c)|c\in C_1\}$.
The permutation equivalent of $\Z_p\Z_{p^2}$-additive codes can be defined similarly.
Then we have the following theorem.

\begin{thm}{\rm\cite{AS15}} \label{th:1}
Let $\mathcal{C}$ be a $\mathbb{Z}_p\mathbb{Z}_{p^2}$-additive code of type $(\alpha,\beta;\gamma,\delta;\kappa)$. Then, $\mathcal{C}$ is permutation equivalent to a $\mathbb{Z}_p\mathbb{Z}_{p^2}$-additive code $\mathcal{C}^\prime$ with generator matrix of the form:
$$\mathcal{G}_S=\left(\begin{array}{cc|ccc}I_\kappa & T' & pT_2 & \mathbf{0} & \mathbf{0}\\
\mathbf{0} & \mathbf{0} & pT_1 & pI_{\gamma-\kappa} & \mathbf{0}\\
\hline
\mathbf{0} & S' & S & R & I_\delta
\end{array}\right),$$
where $I_\delta$ is the identity matrix of size $\delta\times \delta$; $T', S', T_1, T_2$ and $R$ are matrices over $\mathbb{Z}_p$; and $S$ is a matrix over $\mathbb{Z}_{p^2}$.
\end{thm}

From Theorem \ref{th:1}, there is a $\mathbb{Z}_p \mathbb{Z}_{p^2}$-additive code $\mathcal{C}$ of type $(\alpha,\beta;\gamma,\delta;\kappa)$ if and only if
\begin{equation}\label{eq:1}
\begin{array}{c}
\alpha,\beta,\gamma,\delta,\kappa\geq 0, \alpha+\beta>0,\\
0<\delta+\gamma \leq \beta+\kappa \text{\rm{ and }} \kappa\leq \rm{min}(\alpha,\gamma).
\end{array}
\end{equation}

The definition of the duality for $\mathbb{Z}_p\mathbb{Z}_{p^k}$-additive codes is shown in \cite{SWD}, which includes $\mathbb{Z}_p\mathbb{Z}_{p^2}$-additive codes. More detail are suggested to see \cite{SWD}.

\section{Rank of $\mathbb{Z}_p\mathbb{Z}_{p^2}$-additive codes}\label{sec:4}
The following formula for the sum of words of $\mathbb{Z}^\alpha_p\times \mathbb{Z}^\beta_{p^2}$ under the Gray map is useful, see, e.g., \cite{SWL}.
\begin{lem}\label{P1}
For all $\mathbf{u}$, $\mathbf{v}\in \mathbb{Z}^\alpha_p\times \mathbb{Z}^\beta_{p^2}$, $\mathbf{u}=(u_1,\ldots,u_{\alpha+\beta})$, $\mathbf{v}=(v_1,\ldots,v_{\alpha+\beta})$, we have
$$\Phi(\mathbf{u}+\mathbf{v})=\Phi(\mathbf{u})+\Phi(\mathbf{v})+\Phi(pP(\mathbf{u},\mathbf{v})),$$
where $P(\mathbf{u},\mathbf{v})=(0,\ldots,0,P(u_{\alpha+1},v_{\alpha+1}),\ldots,P(u_{\alpha+\beta},v_{\alpha+\beta}))$ and
$$pP(u_i,v_i)=pP(u'_i,v'_i)=\left\{
                            \begin{array}{ll}
                              p & \hbox{if $u'_i+v'_i\geq p$} \\
                              0 & \hbox{otherwise,}
                            \end{array}
                          \right.
u_i=u''_ip+u'_i,\ v_i=v''_ip+v'_i.
$$
\end{lem}
The following expression gives a more useful method to calculate the value of $pP(\mathbf{u},\mathbf{v})$.
\begin{lem}\label{P2}
For any $u$, $v\in \mathbb{Z}_p$ or $u$, $v\in \mathbb{Z}_{p^2}$, we have
$$pP(u,v)=p\sum_{k=1}^{p-1}\sum_{l=p-k}^{p-1}\frac{\prod_{m=0}^{p-1}(u-m)(v-m)}{(u-k)(v-l)},$$
where $P(*,*)$ is defined in Lemma \ref{P1}.
\end{lem}
\begin{proof}
It is a straightforward result from the Lagrange interpolation formula in \cite{Lac} with all interpolation points $(u,v)$, by Lemma \ref{P1}.
\end{proof}
By Lemma \ref{P2}, it is easy to observe that $pP(u,v)$ is a combination of the terms $u^iv^j(0\leq i,j\leq p-1)$. Next, we will figure out that the coefficient of each term
equal to zero or not.

Let $x_j=j$ with $0\leq j \leq p-1$. Denote $\sum\limits_{0\leq i_1\leq\cdots \leq i_j\leq p-1}x_{i_1}\cdots x_{i_j}$ by $\sigma_j(x_0,x_1,\ldots,x_{p-1})$ with $1\leq j \leq p-1$. For simplicity, denote $\sigma_j(x_0,x_1,\ldots,x_{p-1})$ by $\sigma_j$. We call that all $\sigma_j$ are symmetric polynomials, since $\sigma_j$ is invariant under any permutation of $x_0$,$\ldots$,$x_{p-1}$. It is easy to check that $\sigma_j\equiv 0$ mod $p$ with $1\leq j\leq p-1$ and $\sigma_{p-1}\equiv 1$ mod $p$. By using $\sigma_j$, define the following symmetric polynomial:
$$\sigma_i(\widehat{x_k})=\sigma_i(x_0,x_1,\ldots,x_{k-1},x_{k+1},\ldots,x_{p-1}),$$
where $0\leq k\leq p-1$. By using $\sigma_i(\widehat{x_k})$, define the following polynomial:
$$s(i,j)=\sum_{k=1}^{p-1}\sum_{l=p-k}^{p-1}(-1)^{i+j}\sigma_i(\widehat{x_k})\sigma_j(\widehat{x_l}),$$
where $0\leq i,j\leq p-2$.
\begin{lem}\label{sij}
For any $i$, $j\in \{0,1,\ldots,p-2\}$, we have
$$s(i,j)\equiv\sum_{k=1}^{p-1}\sum_{l=p-k}^{p-1}k^il^j \ \mbox{mod}\ p.$$
Moreover, let $T=\{(i,j)|i+j=p-1,1\leq i\leq p-2\}\cup\{(i,j)|i+j=p+2(m-1),0\leq m\leq \frac{p-3}{2},2m\leq i\leq p-2\}$, then
$$s(i,j)\ \mbox{mod}\ p\left\{
          \begin{array}{ll}
            \neq 0 & \hbox{if $(i,j)\in T$} \\
            = 0& \hbox{otherwise}
          \end{array}
        \right..
$$
\end{lem}
\begin{proof}
It is enough to prove that $\sigma_j(\widehat{x_l})\equiv (-l)^j$ mod $p$ for $0\leq j\leq p-2$ and $1\leq l\leq p-1$. When $j=0$, the equation holds. When $1\leq j\leq p-2$, the equation holds from $\sigma_j=\sigma_j(\widehat{x_l})+x_l\sigma_{j-1}(\widehat{x_l})\equiv 0$ mod $p$ and $\sigma_1(\widehat{x_l})\equiv-x_l$ mod $p$.

If $i=j=0$, then $s(0,0)\equiv\sum_{k=1}^{p-1}\sum_{l=p-k}^{p-1}1=\sum_{k=1}^{p-1}k\equiv 0$ mod $p$. If $i=0$ and $j\neq 0$, then $s(0,j)\equiv\sum_{k=1}^{p-1}\sum_{l=p-k}^{p-1}l^j=\sum_{l=1}^{p-1}l^{j+1}$ mod $p$. It is easy to see that $s(0,j)\equiv 0$ mod $p$ with $1\leq j\leq p-3$ and $s(0,p-2)\equiv p-1$ mod $p$. Similarly, $s(i,0)\equiv 0$ mod $p$ with $1\leq i\leq p-3$ and $s(p-2,0)\equiv p-1$ mod $p$. By Bernoulli numbers in \cite{Bn}, we can discuss the other cases.
\end{proof}
\begin{lem}\label{P3}
For any $u$, $v\in \mathbb{Z}_{p^2}$, we have
\begin{eqnarray*}
 pP(u,v) &=& p\sum_{k=1}^{p-1}\sum_{l=p-k}^{p-1}\frac{\prod_{m=0}^{p-1}(u-m)(v-m)}{(u-k)(v-l)}\\
&=& p\bigg(\sum_{i=0}^{p-2}s(i,p-2-i)u^{p-1-i}v^{i+1}+\sum_{i=1}^{p-2}s(i,p-1-i)u^{p-1-i}v^{i}\\
&&+\sum_{j=1}^{\frac{p-3}{2}}\big(\sum_{i=2j}^{p-2}s(i,p+2j-i-2)u^{p-1-i}v^{i-2j+1}\big)\bigg).
\end{eqnarray*}
\end{lem}
\begin{proof}
The result follows from Lemma \ref{sij}.
\end{proof}

Let $\mathcal{C}$ be a $\mathbb{Z}_p\mathbb{Z}_{p^2}$-additive code of type $(\alpha,\beta;\gamma,\delta;\kappa)$ and $C=\Phi(\mathcal{C})$ of length $\alpha+p\beta$. We have known the definition of $rank(C)$ before. In this section, for general prime $p$, we give the range of values $r=rank(C)$ and prove that there is a $\mathbb{Z}_p\mathbb{Z}_{p^2}$-linear code of type $(\alpha,\beta;\gamma,\delta;\kappa)$ with $r=rank(C)$ for any positive integer $r$.

For $\mathbf{u}$, $\mathbf{v}\in \mathbb{Z}^\alpha_p\times \mathbb{Z}^\beta_{p^2}$, denote $(u_1v_1,\ldots,u_\alpha v_\alpha,u_{\alpha+1}v_{\alpha+1},\ldots,u_{\alpha+\beta}v_{\alpha+\beta})$ by $\mathbf{u}*\mathbf{v}$. Let $\mathbf{u}^m=\underbrace{\mathbf{u}*\cdots*\mathbf{u}}_m$, where $m$ is a positive integer.
\begin{thm}\label{base}
Let $\mathcal{C}$ be a $\mathbb{Z}_p\mathbb{Z}_{p^2}$-additive code of type $(\alpha,\beta;\gamma,\delta;\kappa)$ which satisfies (1), $C=\Phi(\mathcal{C})$ be the corresponding $\mathbb{Z}_p\mathbb{Z}_{p^2}$-linear code of length $n=\alpha+p\beta$.
\begin{enumerate}
\item [(i)] Let $\mathcal{G}$ be the generator matrix of $\mathcal{C}$, and let $\{\mathbf{u}_i\}_{i=1}^\gamma$, $\{\mathbf{v}_j\}_{j=1}^\delta$ be the rows of order $p$ and $p^2$ in $\mathcal{G}$, respectively. Denote $\{1+2m|1\leq m\leq \frac{p-1}{2}\}\cup \{p-1\}$ by $L$. Then, $\langle C\rangle$ is generated by $\{\Phi(\mathbf{u}_i)\}_{i=1}^\gamma$, $\{\Phi(\mathbf{v}_j)\}_{j=1}^\delta$ and $\{\Phi(p\mathbf{v}_{i_1}*\cdots*\mathbf{v}_{i_l})\}_{1\leq i_1\leq\cdots\leq i_l\leq \delta}$ with $l\in L$.
\item [(ii)] $rank(C)\in \left\{\gamma+2\delta,\ldots,\min\bigg(\beta+\gamma+\kappa,\gamma+\delta+\sum_{l\in L}{\delta+l-1 \choose l}\bigg)\right\}$. Let $rank(C)=r=\gamma+2\delta+\overline{r}$. Then $\overline{r}\in \left\{0,1,\ldots,\min\bigg(\beta-(\gamma-\kappa)-\delta,\sum_{l\in L}{\delta+l-1 \choose l}-\delta\bigg)\right\}$.
\item [(iii)] The linear code $\langle C\rangle$ over $\mathbb{Z}_p$ is $\mathbb{Z}_p\mathbb{Z}_{p^2}$-linear.
\end{enumerate}
\end{thm}
\begin{proof}

\begin{enumerate}
\item [(i)] Let $\mathbf{c}\in \mathcal{C}$, without loss of generality, $\mathbf{c}$ can be expressed as $\mathbf{c}=\sum_{j=1}^\zeta \mathbf{v}_j+\omega$, where $\zeta\leq \delta$ and $\omega$ is a codeword in $\mathcal{C}$ of order $p$. By Lemma \ref{P1}, we have $\Phi(\mathbf{c})=\Phi(\sum_{j=1}^\zeta \mathbf{v}_j)+\Phi(\omega)$, where $\Phi(\omega)$ is a linear combination of $\{\Phi(\mathbf{u}_i)\}_{i=1}^\gamma$, $\{\Phi(p\mathbf{v}_j)\}_{j=1}^\delta$. And $\Phi(\sum_{j=1}^\zeta \mathbf{v}_j)=\sum_{j=1}^\zeta\Phi(\mathbf{v}_j)+\sum_{l\in L}\sum_{1\leq i_1\leq\cdots\leq i_l\leq \zeta}(\tilde{a}_l\Phi(p\mathbf{v}_{i_1}*\cdots*\mathbf{v}_{i_l}))$, where $\tilde{a}_l\in \mathbb{Z}_p$ and $L=\{1+2m|1\leq m\leq \frac{p-1}{2}\}\cup \{p-1\}$. By Lemma \ref{P3}, $\Phi(\mathbf{c})$ is generated by $\{\Phi(\mathbf{u}_i)\}_{i=1}^\gamma$, $\{\Phi(\mathbf{v}_j)\}_{j=1}^\delta$ and $\{\Phi(p\mathbf{v}_{i_1}*\cdots*\mathbf{v}_{i_l})\}_{1\leq i_1\leq\cdots\leq i_l\leq \delta}$ with $l\in L$.
\item [(ii)] The bound $\gamma+\delta+\sum_{l\in L}{\delta+l-1 \choose l}$ is straightforward by $(i)$, and the $\beta+\gamma+\kappa$ is trivial by the form of $\mathcal{G}_S$ in Theorem \ref{th:1}.
\item [(iii)] Let $\mathcal{S}_{\mathcal{C}}$ be a $\mathbb{Z}_p\mathbb{Z}_{p^2}$-additive code generated by $\{\mathbf{u}_i\}_{i=1}^\gamma$, $\{\mathbf{v}_j\}_{j=1}^\delta$ and $\{p\mathbf{v}_{i_1}*\cdots*\mathbf{v}_{i_l}\}_{1\leq i_1\leq\cdots\leq i_l\leq \delta}$ with $l\in L$. It is obvious that $\mathcal{S}_{\mathcal{C}}$ is of type $(\alpha,\beta;\gamma+\overline{r},\delta;\kappa)$ and $\Phi(\mathcal{S}_{\mathcal{C}})=\langle C\rangle$.
\end{enumerate}
\end{proof}

Let $a,b$ be two positive integers. Denote combinatorial number by ${a \choose b}$ whose value is $\frac{a!}{b!(a-b)!}$ if $a\geq b$ and is zero if $a<b$, respectively.
\begin{lem}\label{nck}
If $1\leq i\leq k$, then we have that
$${k \choose i}=\sum_{j=i-1}^{k-1}{j \choose i-1}.$$
\end{lem}
\begin{proof}
It is easy to check it by ${k \choose i}={k-1 \choose i-1}+{k-1 \choose i}$, where $1\leq i\leq k$.
\end{proof}
The following result is straightforward by Lemma \ref{nck}.
\begin{lem}\label{nckk}
The number $\sum_{l\in L}{\delta+l-1 \choose l}=\sum_{m=0}^{\frac{p-1}{2}}{\delta+2m \choose 2m+1}+{\delta+p-2 \choose p-1}$ can be represented as following
$$\sum_{i=0}^{p-2}\bigg(\sum_{m=0}^{\frac{p-1-\overline{i}}{2}}{\overline{i}+2m \choose i}+{p-2 \choose i}\bigg){\delta \choose i+1}+{\delta \choose p},$$
where $\overline{i}=\left\{
                       \begin{array}{ll}
                         i, & \hbox{i is even;} \\
                         i+1, & \hbox{i is odd.}
                       \end{array}
                     \right.$
\end{lem}

According to Lemma \ref{nckk}, we can give a construction method to construct a $\mathbb{Z}_p\mathbb{Z}_{p^2}$-linear code for each value of $rank(C)$ which is shown before in the following theorem.
\begin{thm}\label{rk}
Let $\alpha,\beta,\gamma,\delta,\kappa$ be positive integers satisfying (\ref{eq:1}).
Then, there is a $\mathbb{Z}_p \mathbb{Z}_{p^2}$-linear code $C$ of type $(\alpha,\beta;\gamma,\delta;\kappa)$ with $rank(C)=r$ if and only if
$$r\in \bigg\{\gamma+2\delta,\ldots,\min\bigg(\beta+\delta+\kappa,\gamma+\delta+\sum_{l\in L}{\delta+l-1 \choose l}\bigg)\bigg\},$$
where $L=\{1+2m|1\leq m\leq \frac{p-1}{2}\}\cup \{p-1\}$.
\end{thm}
\begin{proof}
Let $\C$ be a $\Z_p\Z_{p^2}$-additive code of type $(\alpha,\beta;\gamma,\delta;\kappa)$ with generator matrix
$$\mathcal{G}=\left(\begin{array}{cc|ccc}I_\kappa & T' & \mathbf{0} & \mathbf{0} & \mathbf{0}\\
\mathbf{0} & \mathbf{0} & pT_1 & pI_{\gamma-\kappa} & \mathbf{0}\\
\hline
\mathbf{0} & S' & S_r & \mathbf{0} & I_\delta
\end{array}\right),$$
where $S_r$ is a matrix over $\Z_{p^2}$ of size $\delta\times(\beta-(\gamma-\kappa)-\delta)$, and $C=\Phi(\C)$ is a $\Z_p\Z_{p^2}$-linear code.
Let $\{\mathbf{u}_i\}_{i=1}^\gamma$ and $\{\mathbf{v}_j\}_{j=1}^\delta$ be the row vectors of order $p$ and $p^2$ in $\mathcal{G}$, respectively.

The necessary condition follows from Theorem \ref{base}.
For the sufficiency of the theorem, we should construct a $\mathbb{Z}_p \mathbb{Z}_{p^2}$-linear code $C$ of type $(\alpha,\beta;\gamma,\delta;\kappa)$ with $rank(C)=r$. Since the bound $\beta+\delta+\kappa$ is trivial, it is enough to consider
$\min\bigg(\beta+\delta+\kappa,\gamma+\delta+\sum_{l\in L}{\delta+l-1 \choose l}\bigg)=\gamma+\delta+\sum_{l\in L}{\delta+l-1 \choose l}$,
that is,
$\overline{r}\in \bigg\{0,1,\ldots,\sum_{l\in L}{\delta+l-1 \choose l}-\delta\bigg\}$.
We construct $S_r$ over $\mathbb{Z}_{p^2}$ in the following way such that the code $\C$ reaches the bound $\gamma+\delta+\sum_{l\in L}{\delta+l-1 \choose l}$.

Let $s\leq \delta$ and $\mathbf{x}=(x_1,x_2,\ldots,x_s)\in \mathbb{Z}^s_{p^2}$, where $x_i\neq 0$ for $1\leq i\leq s$. Let $M^\delta(\mathbf{x})$ be a $\delta\times {\delta \choose s}$ matrix with ${\delta \choose s}$ different columns whose $i$-th nonzero component of each column is $x_i$ for $1\leq i\leq s$. For $1\leq j\leq p-1$, define the following matrices of size $j\times 2$
$$A^j_1=\left(
          \begin{array}{cc}
            p-j & p-j \\
            p-j & p-j+1 \\
            \vdots & \vdots \\
            p-j & p-1 \\
          \end{array}
        \right).
$$
For given $j$ and $2\leq i\leq j$, define the following matrices of size ${j \choose i}\times (i+1)$ with matrix $J^j_i$ of size ${j \choose i}\times (i+1)$, whose every entry is $1$,
$$A_i^j=\left(
\begin{array}{c|@{}c@{}}  \mathbf{a}_{p-j+1} & \begin{array}{c}A_{i-1}^{i-1}\\ A_{i-1}^{i}\\\vdots\\A_{i-1}^{j-1}\end{array}
\end{array}
\right)-J^j_i,
$$
where $\mathbf{a}_{p-j+1}$ is the first column of $A^j_i$ and each entry of that is $p-j+1$.

When $i=0$, the coefficient of ${\delta \choose 1}$ in Lemma \ref{nckk} is $\frac{p+3}{2}$, then we consider the following $\frac{p-1}{2}$ matrices
$$\{2I_\delta,3I_\delta,\ldots,\frac{p+1}{2}I_\delta\}.$$
When $1\leq i\leq p-2$, the coefficient of ${\delta \choose i+1}$ in Lemma \ref{nckk} is $\bigg(\sum_{m=0}^{\frac{p-1-\overline{i}}{2}}{\overline{i}+2m \choose i}+{p-2 \choose i}\bigg)$. Denote $\Gamma_i$ by the set of all rows are form $A_i^{\overline{i}+2m}$ with $0\leq m\leq \frac{p-1-\overline{i}}{2}$ and $A_i^{p-2}$, whose size is $\bigg(\sum_{m=0}^{\frac{p-1-\overline{i}}{2}}{\overline{i}+2m \choose i}+{p-2 \choose i}\bigg)$. Then we consider the following matrices
$$\{M^\delta(\mathbf{x})| \mathbf{x}\in \Gamma_i\}.$$
When $i=p-1$, the coefficient of ${\delta \choose p}$ in Lemma \ref{nckk} is $1$, then we consider the following matrix
$$M^\delta(\mathbf{1}), \mathbf{1}=(\underbrace{1,1,\ldots,1}_{p\ times}).$$
With matrices constructed above, we claim two facts as following.
\begin{enumerate}
\item [(1)] Let $S_r$ be composed of $\{2I_\delta,3I_\delta,\ldots,\frac{p+1}{2}I_\delta\}$, $\{M^\delta(\mathbf{x})| \mathbf{x}\in \Gamma_i\}_{i=1}^{p-2}$ and $M^\delta(\mathbf{1})$. When the generator matrix $\mathcal{G}$ of the code $C$ hsa such $S_r$, $rank(C)$ reaches the bound $\gamma+\delta+\sum_{l\in L}{\delta+l-1 \choose l}$.
\item [(2)] Let $S'_r$ be a submatrix of $S_r$ by removing any $t$ columns from the $S_r$, where $1\leq t\leq \sum_{l\in L}{\delta+l-1 \choose l}-\delta$. When the generator matrix $\mathcal{G}$ of the code $C$ has such $S'_r$, $rank(C)=r$ and $r=\gamma+\delta+\sum_{l\in L}{\delta+l-1 \choose l}-t$.
\end{enumerate}

$\langle C\rangle$ is generated by $\{\Phi(\mathbf{u}_i)\}_{i=1}^\gamma$, $\{\Phi(\mathbf{v}_j)\}_{j=1}^\delta$ and $\{\Phi(p\mathbf{v}_{i_1}*\cdots*\mathbf{v}_{i_l})\}_{1\leq i_1\leq\cdots\leq i_l\leq \delta}$ with $l\in L$ by Theorem \ref{base}. We divide these vectors of $\langle C\rangle$ into the following four parts:
\begin{enumerate}
\item [(i)] $\{\Phi(\mathbf{u}_i)\}_{i=1}^\gamma$, $\{\Phi(\mathbf{v}_j)\}_{j=1}^\delta$ and $\{\Phi(p\mathbf{v}^{p-1}_j)\}_{j=1}^\delta$;
\item [(ii)] $\{\Phi(p\mathbf{v}_j^l)\}_{j=1}^\delta$ with $l\in \{1+2m|1\leq m\leq \frac{p-1}{2}\}$;
\item [(iii)] For $1\leq j\leq p-2$, $\{\Phi(p\mathbf{v}^{a_1}_{i_1}*\cdots *\mathbf{v}^{a_{j+1}}_{i_{j+1}})\}_{1\leq i_1<\cdots < i_{j+1}\leq \delta}$ with $a_1+\cdots+a_{j+1}\in L$;
\item [(iv)] $\{\Phi(p\mathbf{v}_{i_1}*\cdots*\mathbf{v}_{i_p})\}_{1\leq i_1<\cdots <i_p\leq \delta}$.
\end{enumerate}
It is easy to check that when $\mathcal{G}$ without $S_r$($S_r=(\mathbf{0})$), the $\langle C\rangle$ is generated by $\{\Phi(\mathbf{u}_i)\}_{i=1}^\gamma$, $\{\Phi(\mathbf{v}_j)\}_{j=1}^\delta$, $\{\Phi(p\mathbf{v}^{p-1}_j)\}_{j=1}^\delta$ and $r=\gamma+ 2\delta$, i.e., $I_\delta$ ensures the vectors $\{\Phi(\mathbf{v}_j)\}_{j=1}^\delta$, $\{\Phi(p\mathbf{v}^{p-1}_j)\}_{j=1}^\delta$ linear independence. It is also easy to check that when $S_r=M^\delta(\mathbf{1})$, the $\langle C\rangle$ is generated by $\{\Phi(\mathbf{u}_i)\}_{i=1}^\gamma$, $\{\Phi(\mathbf{v}_j)\}_{j=1}^\delta$, $\{\Phi(p\mathbf{v}^{p-1}_j)\}_{j=1}^\delta$, $\{\Phi(p\mathbf{v}_{i_1}*\cdots*\mathbf{v}_{i_p})\}_{1\leq i_1<\cdots <i_p\leq \delta}$ and $r=\gamma+ 2\delta+{\delta \choose p}$, i.e., $M^\delta(\mathbf{1})$ ensures the vectors $\{\Phi(p\mathbf{v}_{i_1}*\cdots*\mathbf{v}_{i_p})\}_{1\leq i_1<\cdots <i_p\leq \delta}$ linear independence. Similarly, for $1\leq m\leq \frac{p-1}{2}$, $(m+1)I_\delta$ ensures the vectors $\{\Phi(p\mathbf{v}_j^{1+2m})\}_{j=1}^\delta$ linear independence. For $1\leq j\leq p-2$, $\{\Phi(p\mathbf{v}^{a_1}_{i_1}*\cdots *\mathbf{v}^{a_{j+1}}_{i_{j+1}})\}_{1\leq i_1<\cdots < i_{j+1}\leq \delta}$ with $a_1+\cdots+a_{j+1}\in L$ are linear independence by $\{M^\delta(\mathbf{x})| \mathbf{x}\in \Gamma_j\}$. So the vectors $\{\Phi(\mathbf{u}_i)\}_{i=1}^\gamma$, $\{\Phi(\mathbf{v}_j)\}_{j=1}^\delta$ and $\{\Phi(p\mathbf{v}_{i_1}*\cdots*\mathbf{v}_{i_l})\}_{1\leq i_1\leq\cdots\leq i_l\leq \delta}$ with $l\in L$ are the basis of $\langle C\rangle$ with $S_r$ and $r$ reaches the bound $\gamma+\delta+\sum_{l\in L}{\delta+l-1 \choose l}$, i.e, the claim $(1)$ is true.

For the claim $(2)$, it is enough to prove that the case $t=1$. Without loss of generality, if we remove the first column from $M^\delta(\mathbf{1})$, then we have $\Phi(p\mathbf{v}_{1}*\cdots*\mathbf{v}_{p})=\mathbf{0}$. Actually, if we remove any one column from $M^\delta(\mathbf{1})$, then there is only one vector from $\{\Phi(p\mathbf{v}_{i_1}*\cdots*\mathbf{v}_{i_p})\}_{1\leq i_1<\cdots <i_p\leq \delta}$ which equals to $\mathbf{0}$. For $1\leq m\leq \frac{p-1}{2}$, if we remove any one column from $(m+1)I_\delta$, then it is easy to check that exist one vector of $\{\Phi(p\mathbf{v}_j^{1+2m})\}_{j=1}^\delta$ can be linearly represented by the other vectors. Similarly, we have the other cases $\{M^\delta(\mathbf{x})|\mathbf{x}\in \Gamma_j\}_{j=1}^{p-2}$.

Based on the claims above, the theorem is proved.
\end{proof}

\section{Pairs of rank and dimension of kernel of $\mathbb{Z}_p\mathbb{Z}_{p^2}$-additive codes}\label{sec:5}
The dimension of the kernel of $\mathbb{Z}_p\mathbb{Z}_{p^2}$-additive code have been studied in \cite[Sect. 4]{SWL}. We list some important results of \cite[Sect. 4]{SWL} as follows.
\begin{lem}\cite[Sect. 4]{SWL}\label{l:9}
Let $\mathcal{C}$ be a $\mathbb{Z}_p \mathbb{Z}_{p^2}$-additive code and $C=\Phi(\mathcal{C})$. Then,
$$K(C)=\{\Phi(\mathbf{u})|\mathbf{u}\in\mathcal{C}\text{ \rm{and} } pP(\mathbf{u},\mathbf{v})\in \mathcal{C}, \forall \mathbf{v}\in\mathcal{C}\}.$$
\end{lem}
It's obvious that all codewords of order $p$ in $\C$ under $\Phi$ belong to $K(C)$,
and $\Phi(\mathbf{u})\notin K(C)$ if and only if there is $ \mathbf{v}\in \mathcal{C}$ such that $pP(\mathbf{u},\mathbf{v})\notin \mathcal{C}$.
\begin{prop}\cite[Sect. 4]{SWL}\label{prop:11}
Let $\C$ be a $\Z_p\Z_{p^2}$-additive code of type $(\alpha,\beta;\gamma,\delta;\kappa)$, with generator matrix $\mathcal{G}$, and $C=\Phi(\C)$ with $ker(C)=\gamma+2\delta-\overline{k}$, where $\overline{k}\in \{1,2,\ldots,\delta\}$.
Then, there exists a set $\{\mathbf{v}_1,\mathbf{v}_2,\ldots,\mathbf{v}_{\overline{k}}\}$ of row vectors of order $p^2$ in $\G$ and $\Phi(\vv_i)\notin K(C)$, such that for all $a_i \in \Z_p$,
$$C=\bigcup_{a_i\in \Z_p}\bigg(K(C)+\Phi\bigg(\sum_{i=1}^{\overline{k}}a_i\mathbf{v}_i\bigg)\bigg).$$
\end{prop}

\begin{thm}\cite[Sect. 4]{SWL}\label{Th:13}
Let $\alpha,\beta,\gamma,\delta,\kappa$ be integers satisfying (\ref{eq:1}), then we have the following two assertions:
\begin{enumerate}
  \item [\rm (i)] Let $C$ be a $\mathbb{Z}_p \mathbb{Z}_{p^2}$-linear code $C$ of type $(\alpha,\beta;\gamma,\delta;\kappa)$,
then $ker(C)=\gamma+2\delta-\overline{k}$, where $\overline{k} \in\{0,1,\ldots,\delta\}$.
  \item [\rm (ii)] If $\overline{k} \in\{0,1,\ldots,\delta\}$,
 then there exists a $\mathbb{Z}_p\mathbb{Z}_{p^2}$-linear code $C$ of type $(\alpha,\beta;\gamma,\delta;\kappa)$ with $ker(C)=\gamma+2\delta-\overline{k}$.
  \end{enumerate}
\end{thm}

With the results above, we give the range of the rank of the $\Z_p\Z_{p^2}$-linear codes if its dimension of the kernel is fixed.
Once a possible pair of values $(r,k)$ is given,
the $\Z_p\Z_{p^2}$-linear code $C$ of type $(\alpha,\beta;\gamma,\delta;\kappa)$ with $r=rank(C)$ and $k=ker(C)$ can also be constructed.
\begin{lem}\label{l:15}
Let $\C$ be a $\Z_p\Z_{p^2}$-additive code of type $(\alpha,\beta;\gamma,\delta;\kappa)$,
and $C=\Phi(\C)$ is the corresponding $\Z_p\Z_{p^2}$-linear code.
If $rank(C)=\gamma+2\delta+\overline{r}$ and $ker(C)=\gamma+2\delta-\overline{k}$, where $\overline{k}\in\{1,\ldots,\delta\}$,
then $$1\leq\overline{r}\leq \sum_{l\in L}{\overline{k}+l-1 \choose l}-\overline{k}.$$
Moreover, $\overline{r}=0$ if $\overline{k}=0$.
\end{lem}
\begin{proof}If $\overline{k}=0$, then $C$ is linear and $\overline{r}=0$.
It is clear that $\overline{r}\geq 1$ if $\overline{k}\geq 1$.
Let $\{\mathbf{u}_i\}_{i=1}^\gamma$, $\{\mathbf{v}_j\}_{j=1}^\delta$ be the row vectors of the matrix $\mathcal{G}$ of order $p$ and $p^2$, respectively.
By Proposition \ref{prop:11}, there exists a set $\{\mathbf{v}_1,\mathbf{v}_2,\ldots,\mathbf{v}_{\overline{k}}\}$ of row vectors of order $p^2$ in the generator matrix $\mathcal{G}$,
such that for all $a_i \in \Z_p$,
$C=\bigcup_{a_i\in \Z_p}\bigg(K(C)+\Phi\bigg(\sum_{i=1}^{\overline{k}}a_i\mathbf{v}_i\bigg)\bigg).$
By Lemmas \ref{P1} and \ref{P3},
$\Phi\bigg(\sum_{i=1}^{\overline{k}}a_i\mathbf{v}_i\bigg)=\Sigma_{j=1}^{\overline{k}}a_i\Phi(\mathbf{v}_j)+\Sigma_{l\in L}\Sigma_{1\leq i_1\leq\cdots\leq i_l\leq \overline{k}}(\tilde{a}\Phi(p\mathbf{v}_{i_1}*\cdots*\mathbf{v}_{i_l}))$, $\tilde{a}=\tilde{a}(i_1,i_2,\ldots,i_l)\in \Z_p$ for all $l\in L$.

Let $\{\mu_i\}_{i=1}^{\gamma+\delta}$ be the rows of $\mathcal{G}$.
In fact, $\mu_i=\mathbf{u}_i$ for $i\in \{1,\ldots,\gamma\}$ and $\mu_{j+\gamma}=\mathbf{v}_j$ for $j\in \{1,\ldots,\delta\}$. For $\mathbf{c}\in \mathcal{C}$, let $ord(\mathbf{c})$ mean the order of the vector $\mathbf{c}$, now consider the following two sets $A$ and $B$,
$$A=\{\mathbf{c}| ord(\mathbf{c})=p\}$$
and
$$\begin{aligned}
B=\big\{&\mathbf{c}=\sum_{i=\gamma+\overline{k}+1}^{\gamma+\delta} a_i\mu_i+\sum_{j=1}^{\gamma+\overline{k}} b_j\mu_j\mid a_i\in \Z_{p^2} \text{ and }\exists a_i \in\Z_{p^2}\backslash p\Z_{p^2};\\
&b_j\in \Z_{p} \text{ for } j\in \{1,\ldots,\gamma\} \text{ and } b_j\in p\Z_{p^2} \text{ for } j\in \{\gamma+1,\ldots,\gamma+\overline{k}\}\big\}.
\end{aligned}$$
By the proof of Theorem \ref{Th:13} (ii) in \cite{SWL}, we already know that $K(C)=A\cup B$. The set $A$ can be generated by $\{\mathbf{u}_i\}_{i=1}^\gamma$, $\{p\mathbf{v}_j\}_{j=1}^\delta$. For any vector $\mathbf{c} \in B$, $\Phi(\mathbf{c})=\Phi(\sum_{i=\gamma+\overline{k}+1}^{\gamma+\delta} a_i\mu_i+\sum_{j=1}^{\gamma+\overline{k}}b_j\mu_j)
=\Phi(\sum_{i=\gamma+\overline{k}+1}^{\gamma+\delta} a_i\mu_i)+\Phi(\sum_{j=1}^{\gamma+\overline{k}}b_j\mu_j)$ by Lemma \ref{P1}.
The order of $\sum_{j=1}^{\gamma+\overline{k}}b_j\mu_j$ is $p$, so $\Phi(\sum_{j=1}^{\gamma+\overline{k}}b_j\mu_j)$
can also be generated by $\{\Phi(\mathbf{u}_i)\}_{i=1}^\gamma$, $\{\Phi(p\mathbf{v}_j)\}_{j=1}^\delta$.
By the definition of $\mu_j$, now we only consider  $\Phi(\sum_{i=\gamma+\overline{k}+1}^{\gamma+\delta} a_i\mu_i)=\Phi(\sum_{j=\overline{k}+1}^{\delta} a_j\mathbf{v}_j), a_j \in\Z_{p^2}$.

The first we will prove that if there exists $k\in\{\overline{k}+1,\ldots,\delta\}$,
then $\Phi(pP(\mathbf{v}_k,\mathbf{v}))$
for all $\mathbf{v} \in\C$
is a linear combination of $\{\Phi(\mathbf{u}_i)\}_{i=1}^\gamma$, $\{\Phi(p\mathbf{v}_j)\}_{j=1}^\delta$.
$k\in\{\overline{k}+1,\ldots,\delta\}$ implies that $\Phi(\mathbf{v}_k)\in K(C)$ from the proof of Proposition \ref{prop:11} in \cite{SWL}, that is, for all $\mathbf{v}\in \C$,
$pP(\mathbf{v}_k,\mathbf{v})\in \C$ by Lemma \ref{l:9}.
All vectors $\{\mathbf{u}_i\}_{i=1}^\gamma$, $\{p\mathbf{v}_j\}_{j=1}^\delta$ and
$\{pP(\mathbf{v}_k,\mathbf{v})|k\in\{\overline{k}+1,\ldots,\delta\}\}$ are of order $p$,
thus $\Phi(pP(\mathbf{v}_k,\mathbf{v}))$ can be represented linearly by $\{\Phi(\mathbf{u}_i)\}_{i=1}^\gamma$ and $\{\Phi(p\mathbf{v}_j)\}_{j=1}^\delta$ for $k\in\{\overline{k}+1,\ldots,\delta\}$.

Without loss of generality, let $\delta-\overline{k}=2$, and $j_1=\overline{k}+1, j_2=\overline{k}+2$.
Then $\Phi(a_{j_1}\mathbf{v}_{j_1}+a_{j_2}\mathbf{v}_{j_2})=\Phi(a_{j_1}\mathbf{v}_{j_1})+\Phi(a_{j_2}\mathbf{v}_{j_2})+\Phi(pP(a_{j_1}\mathbf{v}_{j_1},a_{j_2}\mathbf{v}_{j_2}))$.
We know $\Phi(pP(a_{j_1}\mathbf{v}_{j_1},a_{j_2}\mathbf{v}_{j_2}))$ is a linear combination of $\{\Phi(\mathbf{u}_i)\}_{i=1}^\gamma$, $\{\Phi(p\mathbf{v}_j)\}_{j=1}^\delta$.
Since $a_{j_1},a_{j_2}\in \Z_{p^2}$, it is easy to check that $\Phi(a_{j_1}\mathbf{v}_{j_1}),\Phi(a_{j_2}\mathbf{v}_{j_2})$ can also be generated by
$\{\Phi(\mathbf{u}_i)\}_{i=1}^\gamma$, $\{\Phi(p\mathbf{v}_j)\}_{j=1}^\delta$.

As a result, $\langle C\rangle$ can be generated by
$\{\Phi(\mathbf{u}_i)\}_{i=1}^\gamma$, $\{\Phi(\mathbf{v}_j)\}_{j=1}^\delta$, $\Sigma_{1\leq i_1\leq\cdots\leq i_l\leq \overline{k}}(\Phi(p\mathbf{v}_{i_1}*\cdots*\mathbf{v}_{i_l}))$ with $l\in L$.
Hence $\overline{r}\leq \sum_{l\in L}{\overline{k}+l-1 \choose l}-\overline{k}.$
\end{proof}

According to Lemma \ref{l:15}, $rank(C)\in\bigg\{\gamma+2\delta,\ldots,\gamma+2\delta+\sum_{l\in L}{\overline{k}+l-1 \choose l}-\overline{k}\bigg\}$,
and we have
$$rank(C)\in \bigg\{\gamma+2\delta,...,\min\bigg(\beta+\gamma+\kappa,\gamma+2\delta+\sum_{l\in L}{\overline{k}+l-1 \choose l}-\overline{k}\bigg)\bigg\}.$$

\begin{thm}\label{th:16}
Let $\alpha,\beta,\gamma,\delta,\kappa$ be positive integers satisfying (\ref{eq:1}).
Then, there is a $\mathbb{Z}_p \mathbb{Z}_{p^2}$-linear code $C$ of type $(\alpha,\beta;\gamma,\delta;\kappa)$ with $ker(C)=\gamma+2\delta-\overline{k}$ and $rank(C)=\gamma+2\delta+\overline{r}$ if and only if
$\overline{k}\in\{1,\ldots,\delta\}$ and
$\overline{r}\in \bigg\{1,\ldots,\min\bigg(\beta-(\gamma-\kappa)-\delta,\sum_{l\in L}{\overline{k}+l-1 \choose l}-\overline{k}\bigg)\bigg\}$
or $\overline{k}=\overline{r}=0$.
\end{thm}
\begin{proof}
Let $\C$ be a $\Z_p\Z_{p^2}$-additive code of type $(\alpha,\beta;\gamma,\delta;\kappa)$ with generator matrix
$$\mathcal{G}=\left(\begin{array}{cc|ccc}I_\kappa & T' & \mathbf{0} & \mathbf{0} & \mathbf{0}\\
\mathbf{0} & \mathbf{0} & \mathbf{0} & pI_{\gamma-\kappa} & \mathbf{0}\\
\hline
\mathbf{0} & S' & S_{r,k} & \mathbf{0} & I_\delta
\end{array}\right),$$
where $S_{r,k}$ is a matrix over $\Z_{p^2}$ of size $\delta\times (\beta-(\gamma-\kappa)-\delta)$,
and $C=\Phi(\C)$ is its corresponding $\Z_p\Z_{p^2}$-linear code.
The necessary condition is clear by Lemma \ref{l:15}.
Then we will show the sufficient condition.

If $\overline{k}=0$, then $C$ is a linear code and $\overline{r}=0$. When $\overline{k}\in\{1,\ldots,\delta\}$ is fixed, it is enough to construct $S_{r,k}$ over $\Z_{p^2}$ such that the code $\mathcal{C}$ reaches the bound $\overline{r}=\sum_{l\in L}{\overline{k}+l-1 \choose l}-\overline{k}$. Then $S_{r,k}$ is constructed in the following way.
Divide $S_{r,k}$ into two matrices $S_{1},S_{2}$, where $S_1$ is of size $\overline{k}\times(\beta-(\gamma-\kappa)-\delta+1)$,
$S_2$ is of size $(\delta-\overline{k})\times(\beta-(\gamma-\kappa)-\delta+1)$ such that
$S_{r,k}=\begin{pmatrix}S_1\\S_2\end{pmatrix}$, and let $S_2=(\mathbf{0})$.
Let the first column of $S_1$ be $\mathbf{c'}_1=(p-1,p-1,\ldots,p-1)^T$,
this can guarantee $ker(C)=\gamma+2\delta-\overline{k}$ by Theorem \ref{Th:13}.
About the remaining columns in $S_1$, they choose all columns from matrices
$\{2I_{\overline{k}},3I_{\overline{k}},\ldots,\frac{p+1}{2}I_{\overline{k}}\}$, $\{M^{\overline{k}}(\mathbf{x})| \mathbf{x}\in \Gamma_i\}_{i=1}^{p-2}$ and $M^{\overline{k}}(\mathbf{1})$
by replacing $\delta$ with $\overline{k}$ in the proof of Theorem \ref{rk}. By using the same argument with the proof of Theorem \ref{rk}, $rank(C)=\gamma+2\delta+\sum_{l\in L}{\overline{k}+l-1 \choose l}-\overline{k}$, i.e., $\overline{r}$ reaches the bound.
\end{proof}

\section{Rank and dimension of kernel of $\mathbb{Z}_5 \mathbb{Z}_{25}$-linear codes}\label{sec:6}
Let $\mathcal{C}$ be a $\mathbb{Z}_5 \mathbb{Z}_{25}$-additive code of type $(\alpha, \beta;\gamma,\delta;\kappa)$ and $C=\Phi(\mathcal{C})$ of length $\alpha+5\beta$. In this section, we will study the rank and kernel dimension of $\mathbb{Z}_5 \mathbb{Z}_{25}$-linear codes $C=\Phi(\mathcal{C})$ with the results above. If $p=5$, for all $\mathbf{u},\mathbf{v}\in \mathbb{Z}_5^\alpha \times \mathbb{Z}_{25}^\beta$,
it is easy to check that $\Phi(\mathbf{u}+\mathbf{v})=\Phi(\mathbf{u})+\Phi(\mathbf{v})+\Phi(5P(\mathbf{u},\mathbf{v}))$, where
\begin{eqnarray*}
 5P(\mathbf{u},\mathbf{v}) &=& 20 \mathbf{u}^4\ast\mathbf{v} + 15 \mathbf{u}^3\ast\mathbf{v}^2  + 15 \mathbf{u}^2\ast\mathbf{v}^3 + 20 \mathbf{u}\ast \mathbf{v}^4\\
 &&+15 \mathbf{u}^3\ast \mathbf{v} + 10 \mathbf{u}^2\ast  \mathbf{v}^2  + 15 \mathbf{u} \ast\mathbf{v}^3\\
 &&+20\mathbf{u}^2\ast  \mathbf{v} + 20 \mathbf{u}\ast \mathbf{v}^2
 \end{eqnarray*}
 by Lemma \ref{P3}. Hence, $5P(u,v)$ has terms of degree $3,4,5$ which $u,v\in \mathbb{Z}_{25}$.
\begin{thm}\label{Z51}
Let $\mathcal{C}$ be a $\mathbb{Z}_5\mathbb{Z}_{25}$-additive code of type $(\alpha, \beta;\gamma,\delta;\kappa)$ which satisfies (\ref{eq:1}),
$C=\Phi(\mathcal{C})$ be the corresponding $\Z_5\Z_{25}$-linear code of length $n=\alpha+5\beta$.

\begin{itemize}
\item[(i)] Let $\mathcal{G}$ be the generator matrix of $\mathcal{C}$,
and let $\{\mathbf{u}_i\}_{i=1}^\gamma$, $\{\mathbf{v}_j\}_{j=1}^\delta$ be the rows of order $5$ and $25$ in $\mathcal{G}$, respectively.
Then $\langle C\rangle$ is generated by $\{\Phi(\mathbf{u}_i)\}_{i=1}^\gamma$, $\{\Phi(\mathbf{v}_j)\}_{j=1}^\delta$, $\{\Phi(5\mathbf{v}_{i_1}*\mathbf{v}_{i_2}*\mathbf{v}_{i_3})\}_{1\leq i_1\leq i_2\leq i_3\leq \delta}$, $\{\Phi(5\mathbf{v}_{i_1}*\cdots*\mathbf{v}_{i_4})\}_{1\leq i_1\leq\cdots\leq i_4\leq \delta}$ and $\{\Phi(5\mathbf{v}_{i_1}*\cdots*\mathbf{v}_{i_5})\}_{1\leq i_1\leq\cdots\leq i_5\leq \delta}$.
\item[$(ii)$] $rank(C)\in \bigg\{\gamma+2\delta,...,\rm{min}\bigg(\beta+\gamma+\kappa,\gamma+\delta+\dbinom{\delta+2}{3}+\dbinom{\delta+3}{4}+\dbinom{\delta+4}{5}\bigg)\bigg\}$.
Let $rank(C)=r=\gamma+2\delta+\overline{r}$. Then $\overline{r}\in\bigg\{0,1,...,\rm{min}\bigg(\beta-(\gamma-\kappa)-\delta,\dbinom{\delta+2}{3}+\dbinom{\delta+3}{4}+\dbinom{\delta+4}{5}-\delta\bigg)\bigg\}$.
\item[$(iii)$] The linear code $\langle C\rangle$ over $\Z_5$ is $\mathbb{Z}_5 \mathbb{Z}_{25}$-linear.
\end{itemize}
\end{thm}
\begin{proof}
Let $p=5$, then $L=\{3,4,5\}$. Therefore, they are straightforward results from Theorem \ref{base}.
\end{proof}

By Theorem \ref{Z51}, $\langle C\rangle$ is generated by $\{\Phi(\mathbf{u}_i)\}_{i=1}^\gamma$, $\{\Phi(\mathbf{v}_j)\}_{j=1}^\delta$, $\{\Phi(5\mathbf{v}_{i_1}*\mathbf{v}_{i_2}*\mathbf{v}_{i_3})\}_{1\leq i_1\leq i_2\leq i_3\leq \delta}$, $\{\Phi(5\mathbf{v}_{i_1}*\cdots*\mathbf{v}_{i_4})\}_{1\leq i_1\cdots\leq i_4\leq \delta}$ and $\{\Phi(5\mathbf{v}_{i_1}*\cdots*\mathbf{v}_{i_5})\}_{1\leq i_1\cdots\leq i_5\leq \delta}$. We can divide these vectors of $\langle C\rangle$ into the following four parts:
\begin{enumerate}
\item [(1)] $\{\Phi(\mathbf{u}_i)\}_{i=1}^\gamma$, $\{\Phi(\mathbf{v}_j)\}_{j=1}^\delta$ and $\{\Phi(5\mathbf{v}^{4}_j)\}_{j=1}^\delta$;
\item [(2)] $\{\Phi(5\mathbf{v}_j^3)\}_{j=1}^\delta$ and $\{\Phi(5\mathbf{v}_j^5)\}_{j=1}^\delta$;
\item [(3)] For $1\leq j\leq 3$, $\{\Phi(5\mathbf{v}^{a_1}_{i_1}*\cdots *\mathbf{v}^{a_{j+1}}_{i_{j+1}})\}_{1\leq i_1<\cdots < i_{j+1}\leq \delta}$ with $a_1+\cdots+a_{j+1}\in \{3,4,5\}$;
\item [(4)] $\{\Phi(5\mathbf{v}_{i_1}*\cdots*\mathbf{v}_{i_5})\}_{1\leq i_1<\cdots <i_5\leq \delta}$.
\end{enumerate}
The number $\gamma+\delta+\dbinom{\delta+2}{3}+\dbinom{\delta+3}{4}+\dbinom{\delta+4}{5}$ can be represented as
$$\gamma+\sum_{m=0}^{2}{\delta+2m \choose 2m+1}+{\delta+3 \choose 4}=\gamma+4{\delta \choose 1}+9{\delta \choose 2}+10{\delta \choose 3}+5{\delta \choose 4}+{\delta \choose 5},$$
by Lemma \ref{nckk}.

The following theorem can give a construction method to construct a $\Z_5\Z_{25}$-linear code for each value of $rank(C)$ which is shown in Theorem \ref{Z51}.
\begin{thm}\label{Z52}
Let $\alpha,\beta,\gamma,\delta,\kappa$ be positive integers satisfying (\ref{eq:1}).
Then, there is a $\mathbb{Z}_5 \mathbb{Z}_{25}$-linear code $C$ of type $(\alpha,\beta;\gamma,\delta;\kappa)$ with $rank(C)=r$ if and only if
$$r\in \bigg\{\gamma+2\delta,\ldots,\rm{min}\bigg(\beta+\delta+\kappa,\gamma+\delta+\dbinom{\delta+2}{3}+\dbinom{\delta+3}{4}+\dbinom{\delta+4}{5}\bigg)\bigg\}.$$
\end{thm}
\begin{proof}
Let $\C$ be a $\Z_5\Z_{25}$-additive code of type $(\alpha,\beta;\gamma,\delta;\kappa)$ with generator matrix
$$\mathcal{G}=\left(\begin{array}{cc|ccc}I_\kappa & T' & \mathbf{0} & \mathbf{0} & \mathbf{0}\\
\mathbf{0} & \mathbf{0} & 5T_1 & 5I_{\gamma-\kappa} & \mathbf{0}\\
\hline
\mathbf{0} & S' & S_r & \mathbf{0} & I_\delta
\end{array}\right),$$
where $S_r$ is a matrix over $\Z_{25}$ of size $\delta\times(\beta-(\gamma-\kappa)-\delta)$, and $C=\Phi(\C)$ is a $\Z_5\Z_{25}$-linear code.
Let $\{\mathbf{u}_i\}_{i=1}^\gamma$ and $\{\mathbf{v}_j\}_{j=1}^\delta$ be the row vectors of order $5$ and $25$ in $\mathcal{G}$, respectively.

By Theorem \ref{rk}, it is enough to construct $S_r$ over $\mathbb{Z}_{25}$ such that the code $\C$ reaches the bound $\gamma+\delta+\dbinom{\delta+2}{3}+\dbinom{\delta+3}{4}+\dbinom{\delta+4}{5}$. The construction is as follows.

When $i=1$ and $\frac{p+1}{2}=3$, then we have $\{2I_\delta,3I_\delta\}$. When $i=2$, $A_1^1=(4,4)$, $A_1^2=\left(
                                    \begin{array}{cc}
                                      3 & 3 \\
                                      3 & 4 \\
                                    \end{array}
                                  \right)
$, $A_1^3=\left(
            \begin{array}{cc}
              2 & 2 \\
              2 & 3 \\
              2 & 4 \\
            \end{array}
          \right)
$ and $A_1^4=\left(
               \begin{array}{cc}
                 1 & 1 \\
                 1 & 2 \\
                 1 & 3 \\
                 1 & 4 \\
               \end{array}
             \right)
$. Since $9={2 \choose 1}+{3 \choose 1}+{4 \choose 1}$, $\Gamma_1$ is a set of all rows from $A_1^2,A_1^3,A_1^4$. So $\Gamma_1=\{(1,1),(1,2),(1,3),(1,4),(2,2),(2,3),$
$(2,4),(3,3),(3,4)\}$ and matrices $\{M^\delta(\mathbf{x})|\mathbf{x} \in \Gamma_1\}$ are as follows
$$\left(
  \begin{array}{cc}
    1 & \cdots \\
    1 & \cdots \\
    \vdots & \ddots \\
  \end{array}
\right)
\left(
  \begin{array}{cc}
    1 & \cdots \\
    2 & \cdots \\
    \vdots & \ddots \\
  \end{array}
\right)\cdots
\left(
  \begin{array}{cc}
    3 & \cdots \\
    4 & \cdots \\
    \vdots & \ddots \\
  \end{array}
\right).
$$
When $i=3$, $A_2^2=(3,3,3)$, $A_2^3=\left(
                                      \begin{array}{ccc}
                                        2 & 2 & 2 \\
                                        2 & 2 & 3 \\
                                        2 & 3 & 3 \\
                                      \end{array}
                                    \right)
$ and $A_2^4=\left(
               \begin{array}{ccc}
                 1 & 1 & 1 \\
                 1 & 1 & 2 \\
                 1 & 1 & 3 \\
                 1 & 2 & 2 \\
                 1 & 2 & 3 \\
                 1 & 3 & 3 \\
               \end{array}
             \right)
$. Since $10={2 \choose 2}+{3 \choose 2}+{4 \choose 2}$, $\Gamma_2$ is a set of all rows from $A_2^2,A_2^3,A_2^4$ and matrices $\{M^\delta(\mathbf{x})|\mathbf{x} \in \Gamma_2\}$ are as follows
$$\left(
    \begin{array}{cc}
      1 & \cdots \\
      1 & \cdots \\
      1 & \cdots \\
      \vdots & \ddots \\
    \end{array}
  \right)
  \left(
    \begin{array}{cc}
      1 & \cdots \\
      1 & \cdots \\
      2 & \cdots \\
      \vdots & \ddots \\
    \end{array}
  \right)\cdots
  \left(
    \begin{array}{cc}
      3 & \cdots \\
      3 & \cdots \\
      3 & \cdots \\
      \vdots & \ddots \\
    \end{array}
  \right).
$$
When $i=4$, $A_3^3=(2,2,2,2)$ and $A_3^4=\left(
                                          \begin{array}{cccc}
                                            1 & 1 & 1 & 1 \\
                                            1 & 1 & 1 & 2 \\
                                            1 & 1 & 2 & 2 \\
                                            1 & 2 & 2 & 2 \\
                                          \end{array}
                                        \right)
$. Since $5={3 \choose 3}+{4 \choose 3}$, $\Gamma_3$ is a set of all rows from $A^3_3$ and $A_3^4$ an matrices $\{M^\delta(\mathbf{x})|\mathbf{x} \in \Gamma_3\}$ are as follows
$$\left(
    \begin{array}{cc}
      1 & \cdots \\
      1 & \cdots \\
      1 & \cdots \\
      1 & \cdots \\
      \vdots & \ddots \\
    \end{array}
  \right)
  \left(
    \begin{array}{cc}
      1 & \cdots \\
      1 & \cdots \\
      1 & \cdots \\
      2 & \cdots \\
      \vdots & \ddots \\
    \end{array}
  \right)\cdots
  \left(
    \begin{array}{cc}
      2 & \cdots \\
      2 & \cdots \\
      2 & \cdots \\
      2 & \cdots \\
      \vdots & \ddots \\
    \end{array}
  \right).
$$
When $i=5$, then we have $A_4^4=(1,1,1,1,1)$ and $M^\delta(\mathbf{1})$. Let $S_r$ be composed of $\{2I_\delta,3I_\delta\}$, $\{M^\delta(\mathbf{x})| \mathbf{x}\in \Gamma_i\}_{i=1}^{3}$ and $M^\delta(\mathbf{1})$.

Using the similar argument with the proof of Theorem \ref{rk}, when the generator matrix $\mathcal{G}$ of the code $C$ with such $S_r$, $rank(C)$ reaches the bound $\gamma+\delta+\dbinom{\delta+2}{3}+\dbinom{\delta+3}{4}+\dbinom{\delta+4}{5}$.
\end{proof}

With the proofs of Theorems \ref{Th:13}, \ref{Z52} and Lemma \ref{l:15}, we have the following results.
\begin{thm}
Let $\alpha,\beta,\gamma,\delta,\kappa$ be integers satisfying (\ref{eq:1}), then we have the following two assertions:
\begin{enumerate}
  \item [\rm (i)] Let $C$ be a $\mathbb{Z}_5 \mathbb{Z}_{25}$-linear code $C$ of type $(\alpha,\beta;\gamma,\delta;\kappa)$,
then $ker(C)=\gamma+2\delta-\overline{k}$, where $\overline{k} \in\{0,1,\ldots,\delta\}$.
  \item [\rm (ii)] If $\overline{k} \in\{0,1,\ldots,\delta\}$,
 then there exists a $\mathbb{Z}_5\mathbb{Z}_{25}$-linear code $C$ of type $(\alpha,\beta;\gamma,\delta;\kappa)$ with $ker(C)=\gamma+2\delta-\overline{k}$.
  \end{enumerate}
\end{thm}

\begin{lem}
Let $\C$ be a $\Z_5\Z_{25}$-additive code of type $(\alpha,\beta;\gamma,\delta;\kappa)$,
and $C=\Phi(\C)$ is the corresponding $\Z_5\Z_{25}$-linear code.
If $rank(C)=\gamma+2\delta+\overline{r}$ and $ker(C)=\gamma+2\delta-\overline{k}$, where $\overline{k}\in\{1,\ldots,\delta\}$,
then $$1\leq\overline{r}\leq \dbinom{\overline{k}+2}{3}+\dbinom{\overline{k}+3}{4}+\dbinom{\overline{k}+4}{5}-\overline{k}.$$
Moreover, $\overline{r}=0$ if $\overline{k}=0$.
\end{lem}

\begin{thm}
Let $\alpha,\beta,\gamma,\delta,\kappa$ be positive integers satisfying (\ref{eq:1}).
Then, there is a $\mathbb{Z}_p \mathbb{Z}_{p^2}$-linear code $C$ of type $(\alpha,\beta;\gamma,\delta;\kappa)$ with $ker(C)=\gamma+2\delta-\overline{k}$ and $rank(C)=\gamma+2\delta+\overline{r}$ if and only if
$\overline{k}\in\{1,\ldots,\delta\}$ and
$\overline{r}\in \bigg\{1,\ldots,\min\bigg(\beta-(\gamma-\kappa)-\delta,\dbinom{\overline{k}+2}{3}+\dbinom{\overline{k}+3}{4}+\dbinom{\overline{k}+4}{5}-\overline{k}\bigg)\bigg\}$
or $\overline{k}=\overline{r}=0$.
\end{thm}
\section{Conclusion}\label{sec:7}
In this paper, we studied the rank of $\Z_p\Z_{p^2}$-linear codes.
Using combinatorial enumeration techniques, we gave the lower and upper bounds of the rank of $\Z_p\Z_{p^2}$-linear codes.
Moreover, we have constructed $\Z_p\Z_{p^2}$-linear codes for each value of the rank.
Finally, we also constructed $\Z_p\Z_{p^2}$-linear codes for pairs of values of rank and the dimension of the kernel. These results can be helpful to the classification of $\Z_p\Z_{p^2}$-linear codes.

\end{document}